\documentclass[preprint,12pt]{elsarticle}
\usepackage[novbox]{pdfsync}
\usepackage{amsfonts,amsmath,amssymb,mathrsfs,url,
esint,subcaption,
comment}
\usepackage{fancyhdr}
\usepackage{float,multirow,lipsum}
\usepackage{array,listings,verbatim}
\usepackage{tabularx}
\usepackage{nccmath}
\usepackage[utf8]{inputenc}
\usepackage{todonotes}

\usepackage{graphicx}
\graphicspath{{figures/}}
\graphicspath{{Notebook/}}
\usepackage[normalem]{ulem}
\usepackage{cleveref}
\usepackage{lineno}
\usepackage{color}
\usepackage{textcomp}
\usepackage{geometry}
\usepackage{enumitem}

\usepackage{graphicx}
\usepackage{epsfig}
\usepackage{wasysym}
\usepackage{empheq}
\usepackage{url}
\usepackage{tikz}
\usetikzlibrary{positioning}
\usetikzlibrary{arrows}
\usetikzlibrary{arrows.meta}
\usetikzlibrary{calc}
\newtheorem{theorem}{Theorem}
\newtheorem{remark}{Remark}
\newtheorem{proposition}{Proposition}
\newtheorem{definition}{Definition}
\newtheorem{corollary}{Corollary}
\newtheorem{example}{Example}
\newtheorem{lemma}{Lemma}
\newtheorem{assumption}{Assumption}
\newtheorem{claim}{Claim}

 \def\eeD{\end{definition}} \def\beD{\begin{definition}}
\def\beR{\begin{remark}} \def\eeR{\end{remark}}
\def\beL{\begin{lemma}} \def\eeL{\end{lemma}}
\def\beC{\begin{corollary}
  }\def\eeC{\end{corollary}}
  \def\beT{\begin{theorem}}\def\eeT{\end{theorem}}
  \def\beP{\begin{proposition}} \def\eeP{\end{proposition}}
\def\beXa{\begin{example}} \def\eeXa{\end{example}}
\def\beA{\begin{assumption}} \def\eeA{\end{assumption}}

\newtheorem{prop}{Proposition}

\newtheorem{ass}{Assumption}
\def\beAs{\begin{ass}
  }
\def\eeAs{\end{ass}}
\def\fn{\symbolfootnote}
\def\NGM{next generation matrix}

\newcommand{\RN}[1]{\textup{\uppercase\expandafter{\romannumeral#1}}}

\def\L{b}

\def\no{\nonumber}\def\bep{\begin{pmatrix}} \def\eep{\end{pmatrix}}

\def\com{compartment}

\providecommand{\pp}[1]{\left[#1\right]} 
\providecommand{\pr}[1]{\left(#1\right)} 

\usepackage{listings}
\usepackage{color}

\definecolor{dkgreen}{rgb}{0,0.6,0}
\definecolor{gray}{rgb}{0.5,0.5,0.5}
\definecolor{mauve}{rgb}{0.58,0,0.82}

 \newcommand*{\Scale}[2][4]{\scalebox{#1}{$#2$}}%
\def\AB{$(A,B)$ Arino-Brauer epidemic models}
\def\sd{\s_{dfe}}\def\rd{r_{dfe}}\def\a{\alpha}\def\va{\overset{\rightarrow}{\a}} 
\def\b{\beta} \def\g{\gamma}    \def\de{\delta}

   \def\bb{\bff \beta} \newcommand{\bff}[1]{{\mbox{\boldmath$#1$}}}
\def\eqr{\eqref}  
\def\no{\nonumber}
\def\L{\Lambda} \def\La{\Lambda}\def\mR{{\mathcal R}_0}\def\nR{\mathcal R}
\def\fr{\frac}
\providecommand{\pp}[1]{\left[#1\right]} 
\providecommand{\pr}[1]{\left(#1\right)}
\def\im{\item}\def\com{compartment}    

\def\vx{\overset{\rightarrow}{x}}\def\vy{\overset{\rightarrow}{y}}

\def\bep{\begin{pmatrix}} \def\eep{\end{pmatrix}}
\def\bev{\begin{vmatrix}} \def\eev{\end{vmatrix}}
\newcommand{\e}{\;\mathsf e}\def\PF{Perron-Frobenius}
\def\bc{\begin{cases}
  }      
\def\ec{\end{cases}}  
  \def\qu{\quad} 

  \newcommand{\beq}{\begin{eqnarray}
    }
\def\eeq{\end{eqnarray}}
   \newcommand{\be}[1]{\begin{equation}\label{#1}}
\newcommand{\ee}{\end{equation}}
\def\bea{\begin{eqnarray*}}
  \def\Prf{{\bf Proof:} }  
\def\eea{\end{eqnarray*}}  \def\la{\label}\def\fe{for example }   \def\ith{it holds that }\def\saty{satisfy}
\def\sats{satisfies}
\newcommand{\s}{\mathsf s}

\renewcommand{\i}{\mathsf i}
\renewcommand{\r}{\mathsf r}
\newcommand{\f}{\mathsf f}
\newcommand{\h}{\mathsf h}
\newcommand{\p}{\mathsf p}
\def\vi{\overset{\rightarrow}{i}} 
\long\def\symbolfootnote[#1]#2{
\begingroup
\def\thefootnote{\fnsymbol{footnote}}\footnote[#1]{#2}
\endgroup}
\def\fn{\symbolfootnote}
\def\NGM{next generation matrix} 
\def\bzn{basic replacement number}
\def\brn{basic reproduction  number}\def\DFE{disease-free equilibrium}
\def\BEN{\begin{enumerate}}  \def\BI{\begin{itemize}}
\def\EEN{\end{enumerate}}   \def\EI{\end{itemize}}
\def\ME{mathematical epidemiology}\def\QED{\hfill {$\square$}\goodbreak \medskip}\def\eeD{\end{defn}} \def\beD{\begin{defn}}
\newcommand{\E}{{\mathbb E}}\def\resp{respectively}

\def\QED{\hfill {$\square$}\goodbreak \medskip}

\renewcommand{\theta}{\vartheta}

\renewcommand{\thefootnote}{\fnsymbol{footnote}}

\numberwithin{equation}{section}

\def\bc{\begin{cases}
  }     
\def\ec{\end{cases}}  
  \def\qu{\quad}

\def\bea{\begin{eqnarray*}}
    
\def\eea{\end{eqnarray*}} \def\la{\label}\def\fe{for example }   \def\ith{it holds that } 
    \def\sats{satisfies}  \def\saty{satisfy}        

\def\I{\infty} 
  \def\T{\widetilde}
  
\def\BEN{\begin{enumerate}}  \def\BI{\begin{itemize}}
\def\EEN{\end{enumerate}}   \def\EI{\end{itemize}} \def\im{\item}   \def\eqr{\eqref}  
\def\no{\nonumber} 

\def\mR{\mathcal R}

\def\g{\gamma}   \def\de{\delta}  \def\b{\beta}

  \def\resp{respectively}   
 \def\eqr{\eqref}

\def\eeD{\end{definition}} \def\beD{\begin{definition}}
\def\beR{\begin{remark}} \def\eeR{\end{remark}}
\def\beL{\begin{lemma}} \def\eeL{\end{lemma}}

      \def\bb{\bff \b}   \def\vi{\vec \i \;}  \def\vx{\; \vec {\mathsf x}}
    \def\va{\vec \alpha } 
    \def\vy{\; \vec {\mathsf y}} 
    \def\beP{\begin{proposition}} \def\eeP{\end{proposition}}
    \def\beC{\begin{claim}} \def\eeC{\end{claim}}

    \def\m0{{\mathcal R}_0}  

    \long\def\symbolfootnote[#1]#2{
\begingroup
\def\thefootnote{\fnsymbol{footnote}}\footnote[#1]{#2}
\endgroup}
\def\fn{\symbolfootnote}\def\bzn{basic replacement number}
\def\brn{basic reproduction  number}\def\DFE{disease-free equilibrium}\def\com{compartment}
\def\det{deterministic }

 \def\brn{basic reproduction number } \def\sats{satisfies }
\newcommand{\figu}[3]{
\begin{figure}[H]
\centering
\includegraphics[scale=#3]{#1}
\caption{#2\label{f:#1}}
\end{figure}
}

\def\nR{\mathcal R}\def\brn{basic reproduction number }
 \def\DFE{disease free equilibrium}\def\rd{\r_{dfe}} \def\sd{\s_{dfe}}
 
\journal{Arxiv}

\begin{document}

\begin{frontmatter}

\title{Explicit mathematical epidemiology results on age renewal kernels and  $R_0$ formulas are often consequences of the rank one property of the \NGM}

 \author[a]{ Florin Avram}
 \affiliation[a]{organization={Departement of Mathematics,University of Pau},
             city={Pau},
             postcode={64000},
             country={France,},
               addressline={ avramf3@gmail.com},
              }
\author[b]{Rim Adenane}
 \affiliation[b]{organization={Departement of Mathematics,Ibn Tofail University},
             city={Kenitra},
             postcode={14000},
             country={Morocco,},
              addressline={rim.adenane9@gmail.com}
            }

 \author[c]{ Dan Goreac}
 \affiliation[c]{organization={School of Mathematics and Statistics, Shandong University},
 city={  Weihai},
 postcode={264209},
 country={  China,},
 addressline={ LAMA, Univ Gustave Eiffel,  UPEM, Univ Paris Est Creteil, France,  dan.goreac@u-pem.fr}
 }
 
    \author[d]{Andrei Halanay}
 \affiliation[d]{organization={Department of mathematics and informatics, Polytechnic University of Bucharest},
             city={Boucharest},             
             country={Romania, },
             addressline={ andrei.halanay@upb.ro}.
            }


\begin{abstract}
A very large class of ODE epidemic models  \eqr{SYRPH} discussed in this paper  enjoys the property of admitting also an integral renewal formulation, with respect to an ``age of infection kernel" $a(t)$   which has a matrix exponential form \eqr{renker}.
We observe first that a very short proof of this fact is available when there is only one susceptible \com, and when its associated ``new infections" matrix has rank one. In this case, $a(t)$ normalized to have integral 1, is precisely the  probabilistic law which governs the time spent in all the ``infectious states associated to the susceptible \com", and the normalization is precisely the \bzn. The Laplace transform (LT) of $a(t)$ is a generalization of the \bzn,  and its structure  reflects the laws of the times spent in each infectious state. Subsequently, we show that these facts admit extensions to processes
with several susceptible classes, provided that all of them have a new infections matrix of rank one.  These results reveal that the ODE epidemic models highlighted below have also interesting probabilistic properties.
\end{abstract}



\end{frontmatter}

\tableofcontents

\section{Introduction}\la{s:ME}

\indent {\bf Motivation}. Mathematical epidemiology has one fundamental law, which identifies, under certain conditions \cite{Van}, the threshold parameter $R_0$ for the stability of the \DFE\ as the Perron-Frobenius eigenvalue of the ``\NGM" (NGM). Furthermore, explicit $R_0$ formulas are also available when  the \NGM\ has  rank one \cite{Arino,AABBGH}. It may be argued that this result is not that important, since just computing the eigenvalues of the \NGM\ with any symbolic CAS will reveal $R_0$.
However, epidemiologic models whose NGM  has  rank one  have a further important property which does not seem to be  known well enough: it is the existence of   probabilistic renewal kernels -- see \cite{Diek18,AABBGH}, which may be associated to the times spent by infectious individuals in  groups of several infectious states.
This property  is very important both since it
applies to a   very large proportion  of  the models used in applied studies of Covid-19, influenza, ILI (influenza like illnesses), etc., and since it suggests new ways
of calibration --see below. Note that this result implies easily that of \cite{Arino,AABBGH}, by integrating the renewal kernel.
 The lack of awareness for these two fundamental results motivated us to first review them here   for the case of one susceptible \com, and to provide extensions to the case of several susceptible classes. Note especially the formulas \eqr{RMat} and \eqr{kerMat},  which seem to  be new.

 {\bf A bird's eye view of mathematical epidemiology}. Mathematical epidemiology may be said to have started with the  celebrated paper ``A contribution to the mathematical
theory of epidemics" \cite{Ker} on the 1906 plague
epidemic in Bombay, which introduced the SIR (susceptible-infected-recovered) model.\fn[5]{Note that this paper considers the renewal equation formulation, and not just the SIR ODE model.} Bacaer shows that better results can be obtained by adding compartments for rats and fleas \cite[(7-11)]{bacaer2012model},  whose importance   had been overlooked in the first study. In this spirit, each of the three fundamental compartments
S,I,R, could be replaced by classes of several compartments, specific to each epidemics, to be revealed by further analyses.

 The most fundamental aspect  of mathematical epidemiology is the existence of at least two possible fixed states: the boundary ``\DFE" (DFE), corresponding to the elimination of all compartments $\vi$ involving sickness, which may be easily found from the system of non-infectious equations with $\vi=0$, and the ``endemic point" which replaces the DFE when elimination of the sickness is impossible (without intervention, such as quarantine, vaccination, etc...). Note this further induces a partition of all the coordinates and the equations into ``infectious" (eliminable), and the others, or ``non-infectious". The non-infectious may further be divided into recovered (output compartments) and
 susceptible (input compartments). The latter    are  very important;  for example, older individuals may be more susceptible than young ones, or viceversa, and  therefore differentiating susceptibles into several groups may be crucial.

The most important result of \ME\ is the ``\brn"\ $\mR$ threshold theorem concerning the stability of the DFE, already encountered in \cite{Ker}.

There are two flavors of \ME\ and two corresponding  formulas for the \brn:\BEN \im One, for ODE models, identifies $\mR$, under conditions specified in  \cite{Diek,Van,Van08},   via a three-steps ``\NGM" (NGM) procedure:
\BEN \im computing a ``pre-NGM" matrix which involves only the {\bf infectious equations}, \im substituting into it the coordinates of the DFE, which is obtained using only {\bf non-infectious equations}, and
\im computing the spectral radius of the resulting NGM matrix.
\EEN

\im The ``non-Markovian/renewal"  approach adds one further crucial aspect to \ME: the change in infectivity as function of age of  the infection at the time when  transmission took place (which was assumed to be exponentially distributed under the previous approach). This factor enters the model via an ``age of infection kernel", and the \brn\ is computed as the integral of  this kernel   \cite{Diek,Diek10}.
\EEN

The restriction to SIR models with \NGM\ of rank-one renders the connection between the two approaches  very elementary  and yields powerful explicit formulas -- see \eqr{RMat} and \eqr{kerMat} below.

{\bf Contents}. Section \ref{s:SIR} recalls the definition of  SIR-PH-FA models. Section \ref{s:ker} computes the age of infection kernel for these models with one susceptible class, when the ``new infections" matrix $B$ is of rank one. Section \ref{s:exa} provides  examples.  Section \ref{s:ext} provides an extension to the case of several susceptible \com s. Section \ref{s:GLCT} discusses the relation of rank one ODE epidemic models to the generalized linear chain trick (GLCT) formalism. Conclusions and further work are sketched in section \ref{s:conc}.

\section{SIR-PH-FA epidemic models \la{s:SIR}}

{\bf The SIR-PH {(phase-type)} epidemic
models} \cite{Riano} are a particular case of the  \AB\ studied in \cite{AAK,AABBGH}, in which there is  only one input
class $\s$ (and, less importantly, only one output
class ${\mathsf r}$). {One may think of this class  of  models as of processes in which the class I has been  replaced by a transient Markov chain, the time of transition of which models the law (distribution) of the total infectious period. The modeling of infectious laws more general than the exponential is an old concern in epidemiology-- see \fe\ \cite{Feng07,Hurtado21} and references therein. Our concern is not only statistical, but also in identifying laws (principles) which hold for large classes of models.}  Assuming one input
class allows decomposing the \brn\ $\mR$, defined  as the expected number of secondary cases produced by a {\bf typical}
infectious individual during its time of infectiousness-- see \cite{Heth}, which serves also as stability threshold of the DFE,  as
\be{R} \mR=\sd \nR, \ee
where $\sd$ is the  number of susceptibles at the DFE. $\nR$ is called \bzn\ (of {\bf one} susceptible individual). These models include a large number of  epidemic models, like for example for COVID and ILI (influenza like illnesses). After further ignoring certain quadratic terms for the varying population model \cite{AABBGH}, we arrive at a {\bf SIR-PH-FA model}, defined by:
\begin{align}
\label{SYRPH}
\vi'(t) &=  \vi (t) \pp{\s(t)  \; B+ A
- Diag\pp{\bff{\de}+\La \bff 1}}:=
 \vi (t) \pp{\s(t) B-  V}  \no\\
 \s'(t) &= - \s(t) \vi(t) \bb+\La -\pr{ \La + \g_s } \s(t)+ \g_r {\mathsf r}(t), \nonumber\\
 \bb&=\bep\bb_1\\  \vdots\\\bb_n\eep, \quad \bb_i=(B \bff 1)_i=\sum_{j} B_{i,j}, i=1,...,n\nonumber\\
{\mathsf r}'(t) &=  \vi (t) \bff a+ \s(t)  \g_s   - (\g_r+\La){\mathsf r}(t), \qu \bff a=(-A) \bff 1.
 \end{align}
 Here,
\BEN
\im $\s(t) \in \mathbb{R_+}$ represents the set of individuals susceptible to be infected (the beginning state).
\im ${\mathsf r}(t) \in \mathbb{R_+} $  models
recovered  individuals (the end state).
\im $\g_r$ gives  {the rate at which recovered individuals lose immunity,}
and  $\g_s$ gives the rate at which individuals are vaccinated (immunized). These two transfers connect directly the beginning and end states (or classes).
\im the row vector $\vi(t) \in \mathbb{R}^n$  represents the set of individuals in different disease states.
\im $\L >0$ is the per individual death rate, and it equals also the global birth rate (this is due to the fact that this is a model for proportions).

\im  $A$ is a $n\times n$ Markovian sub-generator matrix which describes transfers between the disease classes. Recall that a Markovian sub-generator
matrix for which each
off-diagonal entry $A_{i,j}$, $i\neq j$, satisfies $A_{i,j}\geq 0$, and such that the row-sums are non-positive,  with  at least one inequality being strict.\fn[4]{Alternatively, $-A$ is a  non-singular M-matrix \cite{Arino}, i.e.
  a real matrix $V$ with  $ v_{ij} \leq 0, \forall i \neq j,$ and having eigenvalues
whose real parts are nonnegative \cite{plemmons1977m}.}

The fact a Markovian sub-generator appears  in  our ``disease equations" suggests that certain probabilistic concepts intervene in our deterministic model, and this is indeed the case--see below. Note also that typical epidemic models \saty\ $A_{i,j} A_{j,i}=0$, $i\neq j$, and so this matrix may be arranged  to be triangular.

\im $\bff \de \in \mathbb{R_+}^n$ is a column vectors giving the death rates caused by the epidemic in the  disease \com s. The matrix $-V$, which combines $A$  and the birth and death rates $\La, \bff \de$, is also a Markovian sub-generator.

\im $ B $ is a $n \times n$ matrix (called sometimes ``new infections" matrix). {We will denote by $\bb$ the vector containing the sum of the entries in each row
of $B$, namely, $\bb= B \bff1$.} Its components $ \bb_i$ represent the
{\bf total force of infection} of the  {disease} class $i$, and $\s(t) \vi(t) \bb$ represents the
total flux which must leave class $\s$.
Finally, each  {entry} $B_{i,j}$, multiplied by $\s$,  represents the
force of infection from the  {disease} class $i$ onto class $j$, and our essential assumption below will be that $B_{i,j}=\beta_i \a_j,$ i.e. that all forces of infection are distributed among the infected classes conforming to the same probability vector $\va=(\a_1,\a_2,...,\a_n)$.

\EEN

\beR The matrices $B$ and $-V$,  intervene in the formulas related to the \NGM\ approach.\eeR

\beR \la{r:PF} \BEN \im   Note  the factorization of the  equation for the diseased \com s $\vi$, which ensures the existence of a fixed point where these \com s vanish, and  implies a representation of $\vi$ in terms of $\s$:
\be{irep} \vi(t)=\vi(0)e^{- t V + B \int_0^t s(\tau) d \tau}=\vi(0)e^{\pp{- t Id + B V^{-1} \int_0^t s(\tau) d \tau} V}. \ee

{In this representation intervenes an essential character of our story, the matrix $B V^{-1},$ which is proportional to the \NGM\ $\sd B V^{-1}$}. A  second representation \eqr{vc} below will allow us to embed our models in the interesting class of distributed delay/renewal models, in the case when
$B$ has rank one.
\im
When $B=\bb \va,$ it holds by the Arino-Brauer formula \cite{Arino} generalized in \cite{AABBGH} that
\be{AB}\mR= \lambda_{PF}(\sd B V^{-1})=\sd \va V^{-1} \bb,\ee
where $\lambda_{PF}$ denotes the Perron-Frobenius eigenvalue (positive eigenvalue of maximum modulus).
 \la{r:PF} Note the proof is very simple. The eigenvector must be $v=\bb$, and plugging $\bb$ in $\sd \bb \va V^{-1} v= \mR v$ yields $ \bb (\sd \va V^{-1} \bb)= \mR \bb$, and  the result.

 \im  We added the   FA (first approximation of an exact model for   proportions) to the name of these models, following \cite{AABBGH}, to differentiate them from their   version with demography.
 \EEN
\eeR

\section{Associated Markovian semi-groups,  age of infection kernels, and  an $\nR$  formula for  SIR-PH-FA models with one susceptible class and $B$ of rank one \la{s:ker}}

We show  here  that when  $B$ has rank one,  SIR-PH-FA models have an associated explicit age of infection kernel, which allows in particular obtaining $R_0$ via an integral. We may say that rank one SIR-PH-FA epidemic models lie in the intersection of the  ODE/Markovian and the non-Markovian/renewal models. Alternatively, they are  precisely the renewal models with a matrix-exponential kernel. The equivalence of the two approaches in  this  simple context is  proved  concisely below; it may also be read between the lines of the wider scope papers \cite{Diek18,Diek22}.

Our attention to this subject was drawn by  formulas on  \cite[pg. 3]{Breda} for the  ``distributed delay/renewal/age of infection kernels" for  particular cases of the SIR and SEIR models. These authors  assign as an exercise to extend their formulas to other models; it turned out later that determining  which models to extend to was part of the exercise. Six years later    the exercise was first solved by  Champredon-Dushoff-Earn \cite{Champredon} for  Erlang-Seir models. We provide below  a further extension to the case of SIR-PH-FA models with $B$ of rank one -- see also \cite[Thm. 2.2]{Diek18}, \cite{Diek22} for related results.

\begin{proposition}\label{p:ren}  Let $\T {i}(t) =\vi(t) \bb$ denote the total force of infection of a  SIR-PH-FA model \eqr{SYRPH} with one susceptible class, without loss of immunity, i.e. $\g_r=0$, so that $r(t)$ does not affect the rest of the system, and with $B= \bb \va$ of rank one.
Then
\BEN
\im  The solutions of the ODE system
\eqr{SYRPH}    \saty\ also an integro-differential ``SI system"  of two scalar equations
\begin{align}\label{SI}
\bc \s'(t)=\Lambda -\pr{ \La + \g_s } \s(t)-\s(t) \T i(t)\\
\T i(t)=\vi(0) e^{-tV}\bb+
\int_0^t s(\tau) \T i(\tau)
a(t-\tau) d\tau,\ec
\end{align}

where \be{renker} a(\tau)=  \va e^{-\tau V} \bb, 
 \ee
with $-V=A
- \pr{Diag\pp{\bff{\de}+\Lambda \bff 1}}$  (it may be checked that this fits the formula on page 3 of \cite{Breda} for SEIR when $\Lambda= 0, \de=0$).\fn[3]{$a(t)$ is called ``age of infection/renewal  kernel; see  \cite{Hees,Brauer05,Breda,DiekHeesBrit,Champredon,Diek18,Diek22} for expositions of this concept.}

\im The \bzn\ $\nR$ has an integral representation
\begin{align}\label{R0I}
\nR=\int_0^\I a(\tau)  d \tau= \int_0^\I \va e^{- \tau V} \bb d\tau =\va\ V^{-1} \; \bb.
\end{align}

\EEN
\end{proposition}

\Prf 1. The non-homogeneous   infectious equations may be transformed into an integral equation by applying the variation of constants formula.
The first step is the solution of the homogeneous part.  Denoting this by  $\Gamma(t)$,  \ith \begin{align}
\label{SYRIGI}
\overset{\rightarrow}\Gamma'(t) &=
- \overset{\rightarrow}\Gamma(t) V \Longrightarrow \overset{\rightarrow}\Gamma(t) =\overset{\rightarrow}\Gamma(0) e^{ t (-V)}.
\end{align}

The variation of constants formula implies then that $\vi(t)$ satisfies \ the integral equation:
\be{vc} \vi(t) = \vi(0) e^{-t V}+
\int_0^t \s(\tau) \vi(\tau)  B e^{-(t-\tau) V} d \tau.\ee

Now in the rank one case $B=\bb \va$, and
 \eqr{vc} becomes
\be{vcb} \vi(t) = \vi(0) e^{-t V}+
\int_0^t \s(\tau) \vi(\tau)  \bb \va e^{-(t-\tau) V} d \tau.
\ee

Finally, multiplying both sides on the right by $\bb$ yields the result.

2. By the ``survival method"\fn[3]{This is a first-principles method, whose rich history  is
described in   \cite{Hees,Diek10}-- see  also \cite[(2.3)]{Champredon}, \cite[(5.9)]{Diek18}.}, $\nR$
 may  be obtained by integrating $\Gamma(t)$ with $\Gamma(0)=\va$.
 A direct proof is also possible by noting that all eigenvalues of the \NGM\ except one are $0$ \cite{Arino,AABBGH}.
\QED

\beR When $\overset{\rightarrow}\Gamma(0)$ is a probability vector, \eqr{SYRIGI} has the interesting probabilistic interpretation of the survival probabilities in the various components  of the semigroup generated by the Metzler/Markovian sub-generator matrix $-V$ (which inherits this property from the phase-type generator $A$). Practically, $\overset{\rightarrow}\Gamma(t)$ will give  the expected fractions of individuals who are still in each compartment at time $t$.
\eeR

\beR We may relate \eqr{vcb} to the age of infection equation of the  distributed delay/renewal model, by noting that \ith
 \be{reneq}\vi(t) =\int_{-\I}^t s(\tau) \T i(\tau)
a(t-\tau) d\tau,\ee provided that   $ \T i(\tau)$ on the interval $(-\I,0]$ is  $\fr{k}{s_0} \de_0(\tau),$ where $\de_0(\tau)$ denotes the generalized Dirac function, and that $\vi(0)=k \va$. This second equation is related to  \cite[(2.7b),(2.8),(2.9)]{Champredon} and \cite[(1)]{Breda}.\fn[4]{In fact, these authors work   with the related \emph{incidence flux}   between the $\s$ and $\vi$ variables
${Incid}:=s\overset{\rightarrow}{i}\mathbf{b},$
 denoted by  $i(t)$ in \cite{Champredon}, and by $F(t)$ in \cite{Breda}.}
 Equations like \eqr{reneq}, called DD (distributed delay) equations appear already in the founding paper \cite{Ker}, which is quite natural. Indeed, if it were known that infections arise precisely $\tau_0$ units of time after a contact, then the second equation of the SI model would involve the Dirac kernel
$a(\tau)=\de_{\tau_0}(\tau)$. But, since the value of $\tau_0$ is never known,
it is natural to replace the  Dirac kernel by a continuous one.
\eeR

\beR \BEN \im
The fact that DD systems can be approximated by ODE systems, {by approximating the delay distribution via one of Erlang, and more generally, of matrix-exponential type}, has long been exploited in the epidemic literature, under the name of "linear chain trick" (which has roots in the Erlangization of queueing theory)-- see \fe\ \cite{Wearing2005,Feng,Wang2017,Diek18,cassidy2018recipe,Hurtado19,Ando,Diek22} for  recent contributions and further references.    The opposite direction however, i.e. the solution  of the exercise in  \cite{Breda} of identifying the kernels associated to
ODE models, seems not to have been resolved in this generality, prior to our paper.

\im  Finally, for DD models, normalizing the kernel by its integral $\mR$ yields the density of  the  ``intrinsic generating interval" for the age of infection \cite{Champredon15,Champredon,Diek22}:
 $g(t)=\fr{a(t)}{\mR}=\fr{\va e^{- \tau V} \bb}{\mR}$ --see \cite[(2.6)]{Champredon}.

\EEN
\eeR

\section{Examples}\la{s:exa} 

\subsection{The S$I^2$R/SAIR/SEIR-FA epidemic model}
We define the S$I^2$R/SAIR/SEIR-FA epidemic model \cite{Van08,RobSti,Ansumali,Ott,AAH} by:
\be{SAIR} 
\Scale[0.9]{\bc
\s'(t)= \Lambda  -\s(t)\pr{\beta_ 2  i_2(t)+\beta_ 1  i_1(t)+\g_s+\Lambda} + \g_r r(t)+\de\s(t)  i_2(t)\\
\bep i_1'(t)&i_2'(t)\eep = \bep i_1(t)&i_2(t)\eep
\pp{\s(t) \bep  \beta_1     &  0 \\
 {\beta_ 2}&  0 \eep+ \bep    -(\g_1   +\Lambda)  &  \g_{1,2} \\
 0&  -\pr{\g_2 +\Lambda +\de}\eep}, \g_{1}=\g_{1,2}+\g_{1,r}
\\
r'(t)=  \g_s \s(t)+ \g_{1,r}   i_1(t)+ \g_2     i_2(t)- (\g_r+\Lambda) r(t). \ec}
\ee
\beR \BEN \im  The classic SEIR model is obtained when $\g_{1,r}=0=\de$. This model maybe viewed as an ``Erlangization" of the SIR model, in the sense described by the following definition.
\beD a. A (generalized) Erlang row is a matrix row which has  one negative element on the main diagonal, followed by its opposite to the right, and in which  all the other elements are zero.

b. A square matrix obtained by adding  Erlang rows  above a square matrix A, while extending the columns of A by 0's, will be called an Erlangization of A.

c. When the extension above involves rows with one negative element on the main diagonal, which is followed to the right by a positive element which is smaller in absolute value, and in which  all the other elements are zero, will be called Coxianization. For example, the SI$^2$R model is a
Coxianization of the SIR model.
\eeD
\im Erlangization and Coxianization are particular cases of the generalized
linear chain trick -- see \cite{Hurtado19}  and the full version of this article. Probabilistically, they amount to preceding a compartment I by another compartment E, such that E may  transit either to I or outside the infectious/Markovian classes.

\EEN
\eeR

This process has been called in previous papers under several names. Besides SEIR, used usually when  $ \beta_1=0$, but also with  $ \beta_1>0$ -- see   \cite{Van08}, we have also SITR~\cite{YangBrauer} (when $b=\delta=\g_r=\g_s=0$).

The system \eqr{SAIR} contains nine parameters, three of which $\de,\g_r$ and $\g_s$ do not change much the essence of the problem, and   are often omitted.
It is an \AB\ with  parameters $\va =\bep 1&0\eep, A=\bep -\g_1  &\g_{1,2}\\ 0&-\g_2   \eep, \bff a= (-A)\bff 1
 =\bep \g_{1,r} \\\g_2   \eep$ and
 \[ \bb=\bep \beta_ 1 \\ \beta_ 2\eep, \; \mbox{so}\; B=\bep \beta_ 1 & 0\\ \beta_ 2 & 0 \eep , \bff \de=\bep 0\\\de\eep, V=\bep    \g_1   +\Lambda  &  -\g_{1,2} \\
 0&  {\g_2 +\La +\de}\eep.\]
The Laplace transform of the age of infection kernel  is:
 \begin{align} \la{aTr} \Hat{a}(s)&= \va (s I+V)^{-1} \bb
 =
 \beta_ 1\frac{1 }{ \left(b+\g_1  +s\right)}+
 \beta_ 2 \frac{\g_{1,2}}{(b+\g_2  +\delta+s ) \left(b+\g_1  +s\right)},
 \end{align}
 and the  Arino \& al. formula yields
$
\nR= \int_0^{\I} a(\tau) d \tau= \frac{\beta_ 1 (b+\g_2  +\delta )+\g_{1,2} \beta_ 2}{(b+\g_2  +\delta ) \left(b+\g_1  \right)}.$
 \beR {\bf Probabilistic interpretations of the matrix exponential}. \eqr{aTr} confirms that the coefficients of $\beta_ i$ in the delay kernel $a(t)$  are the components $\vec \Gamma_{1}(t),\vec \Gamma_{2}(t)$ of the semigroup starting from the first state $\vec \Gamma(t)=(1,0) e^{-t V}$, namely an exponential with parameter $b+\g_1$, corresponding to surviving in the exposed/asymptomatic period, and a hypoexponential survival function
corresponding to surviving in the infectious state.
\eeR

\figu{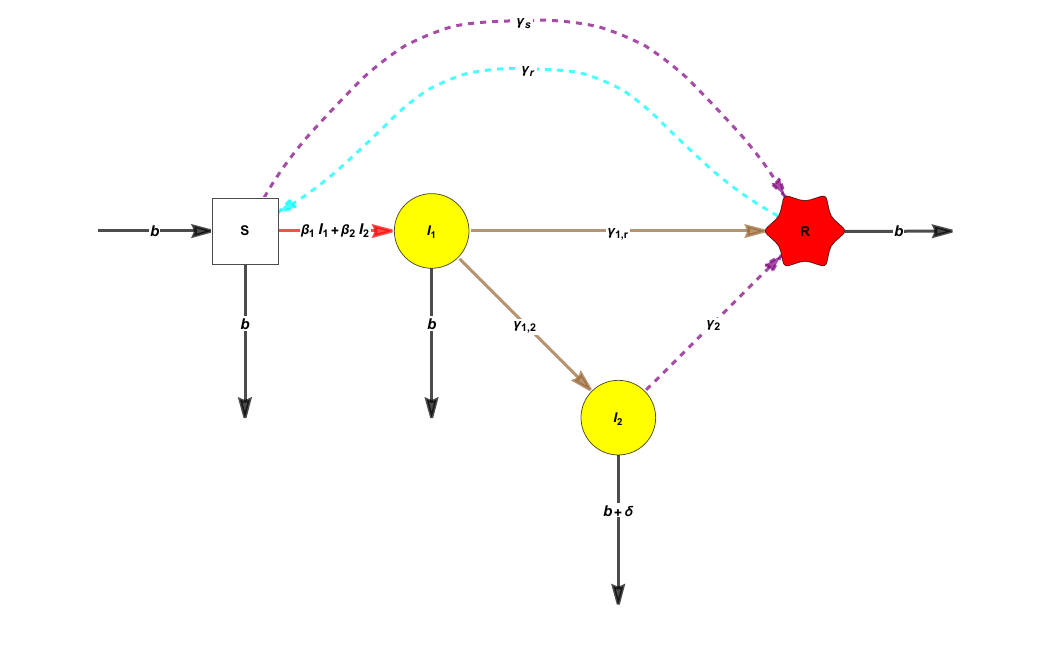}{ Chart flow of the SAIR model \eqr{SAIR}. The red edge corresponds
to the entrance of susceptibles into the disease classes, the brown edges are the rate of the transition matrix V, and the cyan
dashed lines correspond to the rate of loss of immunity.  The remaining black lines correspond to the
inputs and outputs of the birth and natural death rates, respectively, which are equal in this case.}{1}

\subsection{A generealized SLAIR epidemic model}
The SLAIR epidemic model \cite{YangBrauer,arino2020simple,AAV} is defined by:
\be{SLIARG}
\Scale[0.9]{ \bc
\s'(t)=   \Lambda -\s(t)\pr{\beta_ 2  i_2(t)+\beta_3 i_3(t)+\Lambda}\\
\bep i_1'(t)&i_2'(t)& i_3'(t)\eep = \bep i_1(t)&i_2(t)& i_3(t)\eep
\pp{\s(t) \bep  0    &  0& 0 \\
 {\beta_ 2}&  0 &0\\
 \beta_ 3&0&0\eep+ \left(
\begin{array}{ccc}
 -\g_1 -\Lambda & \g_{1,2} & \g_{1,3} \\
 0 & -\g_2 -\Lambda  & \g_{2,3} \\
 0 & 0 & -\g_3 -\Lambda \\
\end{array}
\right)}
\\
r'(t)= \g_{2,r}  i_2(t)+\g_3  i_3(t) -\Lambda r(t)\ec.}
\ee
This is an  \AB\ with  parameters
 \bea
 \Scale[0.9]{ \va =\bep 1&0&0\eep, A= \left(
\begin{array}{ccc}
 -\g_1  & \g_{1,2} & \g_{1,3} \\
 0 & -\g_2   & \g_{2,3}\\
 0 & 0 & -\g_3   \\
\end{array}
\right), \bff a= (-A) \bff 1
 =\bep 0\\ \g_{2,r}  \\\g_3  \eep, \bb=\bep 0\\ \beta_2 \\ \beta_ 3\eep, \;   \mbox{so}\; B=\left(
\begin{array}{ccc}
 0 & 0 & 0 \\
 \beta _2 & 0 & 0 \\
 \beta _3 & 0 & 0 \\
\end{array}
\right).}
\eea
The Laplace transform of the age of infection kernel  is:
 \bea
\Scale[1.1]{ \Hat{a}(s)=\beta _2 \frac{ \gamma _{\text{1,2}}}{\left(b+\gamma _1+s\right) \left(b+\gamma _2+s\right)}+ \beta_ 3 \pr{\frac{\g_{1,3}}{\left(b+\gamma _1+s\right) \left(b+\gamma _3+s\right)}+\frac{\gamma _{\text{1,2}} \gamma _{\text{2,3}}}{\left(b+\gamma _1+s\right) \left(b+\gamma _2+s\right) \left(b+\gamma _3+s\right)}},}\eea
and the Arino \& al. formula yields
$
\nR= \frac{\beta _3 \gamma _{\text{1,2}} \gamma _{\text{2,3}}+b \beta _2 \gamma _{\text{1,2}}+\beta _2 \gamma _3 \gamma _{\text{1,2}}+b \beta _3 \gamma _{\text{1,3}}+\beta _3 \gamma _2 \gamma _{\text{1,3}}}{\left(b+\gamma _1\right) \left(b+\gamma _2\right) \left(b+\gamma _3\right)}.$
\figu{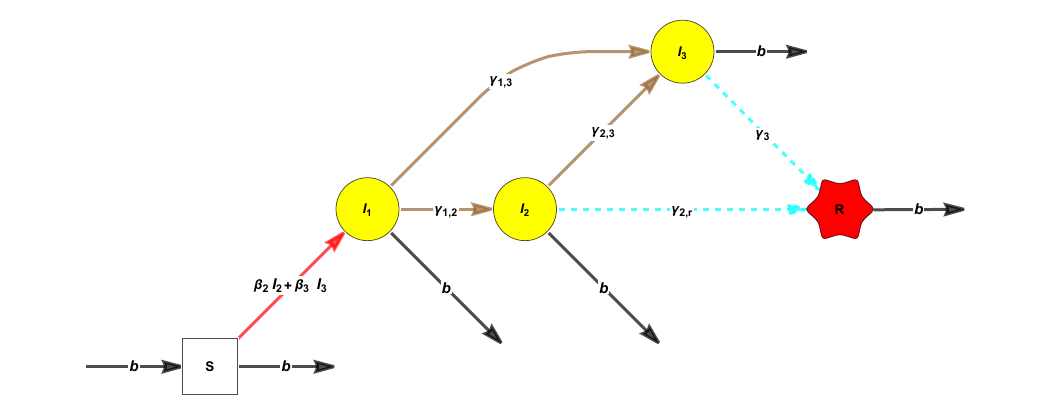}{ Chart flow of the SLAIR model \eqr{SLIARG}.
}{1} 

\section{Extension to several susceptible \com s and time-dependent inputs} \la{s:ext}
{We show here that SIR-PH epidemic models with two or more susceptible \com s  may also \saty\  renewal type integro differential equations}.

Consider  the ODE model with $n$ infectious classes $\vi=(i_1,i_2, ...,i_n),$   $2$ susceptible classes with arrivals $\La_1(t), \La_1(t)$, and total arrivals $\La(t)=\La_1(t)+\La_2(t),$ defined by:
\begin{align}
\label{SIR2}
\s_1'(t) &=\La_1(t) + \g_{r,1} {\mathsf r}(t) - \s_1(t) \pp{\vi(t) {\bb_1} + \La(t) + \g_s }, \; \; \bb_1=B_1 \bff 1,
 \nonumber\\
 \s_2'(t) &= \La_2(t)+ \g_{r,2} {\mathsf r}(t)- \s_2(t) \pp{\vi(t){\bb_2}+  \La(t) + \g_s }, \; \; \bb_2=B_2 \bff 1,
 \nonumber\\
\vi'(t) &=  \vi (t) \pp{\s_1(t)  \; B_1+\s_2(t)  \; B_2+ A
- Diag\pr{ \La(t) \bff 1+\bff{\de}}}:=
 \vi (t) \;  F(t)- \vi (t) V, \no 
 \\
{\mathsf r}'(t) &=  \vi (t) \bff a+ (\s_1(t)+\s_2(t))  \g_s   - (\g_r+\La(t)){\mathsf r}(t), \; \bff a=(-A) \bff 1,
\end{align}
where we put $ F(t)=\sum_i \s_i(t)  B_i,  \bff \de=\bep\de_1\\\de_2\\ \vdots\eep, \g_{r}=\g_{r,1}+\g_{r,2}$. Note that
\[N(t)=\s_1 (t)+\s_2 (t)+{\mathsf r}(t)+ \vi(t) \bff 1 \text{ \sats }\; N'(t)= \La(t) (1-N(t)) - \vi(t) \bff \de.\]

Assume further that $B_i=\bb_i \va_i, i=1,2$,
 note the factorization \be{facB} F(t)
 ={\bb} \bep \s_1(t)&0\\0&\s_2(t)\eep {\va}, \text{  where }\bb=\bep \bb_1,&\bb_2 \eep, {\va} =\bep  \va_1\\ \va_2\eep\ee
 are $n\times 2$ and $2\times n$ matrices, \resp.

 The variation of constants formula applied to \eqr{SIR2} implies  that $\vi(t)$ satisfies \ the integral equation:
\beq \la{vc2}\no  \vi(t) &=&\vi(0) e^{-t V}+
\int_0^t  \vi(\tau)F(\tau) e^{-(t-\tau) V} d \tau =\vi(0) e^{-t V}+
\int_0^t  \vi(t-\tau)F(t-\tau) e^{-\tau V} d \tau\\:&=&\vi(0) e^{-t V}+
\int_0^t  \vi(t-\tau)\T K(t,\tau)  d \tau, \quad \T K(t,\tau):=F(t-\tau) e^{-\tau V}.\eeq
We will call $\T K(t,\tau)$ implicit kernel, to emphasize the fact that it depends still on the unknown $\s_i(t)$.

When $B_i=\bb_i \va_i, i=1,2$, putting $\T {i}_i(t)=\vi(t) \bb_i, i=1,2$
 \eqr{vc2} becomes
\be{vcb2} \vi(t) = \vi(0) e^{-t V}+
\int_0^t \pp{\sum_{i=1}^2\s_i(\tau) \T {i}_i(\tau) \va_i e^{-(t-\tau) V}} d \tau.
\ee

Multiplying further by $\bb_k, k=1,2$ and putting
\be{a}a_{i,j}(t)=  \va_i e^{-t V} \bb_j, i,j=1,2,\ee yields a system of two equations for $\T {i}_k(t)=\vi(t) \bb_k, k=1,2$:
\be{ksys}\bc \T {i}_1(t)=\vi(0) e^{-t V} \bb_1+
\int_0^t \pp{\s_1(\tau) \T {i}_1(\tau) a_{1,1}(t-\tau)+
\s_2(\tau) \T {i}_2(\tau) a_{2,1}(t-\tau)}  d \tau\\
\T {i}_2(t)=\vi(0) e^{-t V} \bb_2+
\int_0^t \pp{\s_1(\tau) \T {i}_1(\tau) a_{1,2}(t-\tau)+
\s_2(\tau) \T {i}_2(\tau) a_{2,2}(t-\tau)} d \tau\ec\ee
We may conclude that proposition \ref{p:ren} extends as follows:
\begin{proposition} Consider a SIR-PH-FA model \eqr{SIR2} with two susceptible classes  with $B_i= \bb_i \va_i, i=1,2$,  with constant input inflows $\Lambda_1,\Lambda_2$, and which satisfies the conditions of \cite{Van}. Then, \ith:
\BEN \im The solutions of \eqr{SIR2}
    \saty\ also an integro-differential  system
\begin{equation} \la{kerMat}
\bc \s_1'(t) =\La_1 + \g_{r,1} {\mathsf r}(t) - \s_1(t) \pp{\vi(t) {\bb_1} + \La + \g_s }, \; \; \bb_1=B_1 \bff 1,\\
 \s_2'(t) = \La_2+ \g_{r,2} {\mathsf r}(t)- \s_2(t) \pp{\vi(t){\bb_2}+  \La + \g_s }, \; \; \bb_2=B_2 \bff 1,\\
{\mathsf r}'(t) =  \vi (t) \bff a+ (\s_1(t)+\s_2(t))  \g_s   - (\g_r+\La){\mathsf r}(t), \; \bff a=(-A) \bff 1\\
\T i(t)=\vi(0) e^{-tV} \bb+
\int_0^t \T i(\tau)
Diag(\s(\tau))  a(t-\tau)  d\tau,\ec
\end{equation}
where $  a(\tau)=  \bep \s_i(\tau) a_{i,j}(\tau)\eep_{i,j=1,2}, \s(\tau)=(\s_1(\tau),\s_2(\tau))$.
\im There is a unique  DFE,  given by
\[\sd^{i}=\fr{\La+ \g_s}{\La_i+\g_{r,i} \rd}, i=1,2, \; \rd=\fr {\g_s}{\La+ \g_r+\g_s}=\fr {\g_s}{\sum_i \La_i+ \sum_i \g_{r,i}+\g_s}.\]

\im  {If $ \bb_1,\bb_2 $ are independent}, then the  \brn\ $\mR$ is the \PF\ eigenvalue of the two by two matrix
\be{RMat} M_D:=\bep \sd^{1} \va_1 V^{-1} \bb_1& \sd^{1} \va_1 V^{-1} \bb_2\\\sd^{2} \va_2 V^{-1} \bb_1&\sd^{2} \va_2 V^{-1} \bb_2\eep={\va}\bep \sd^{1}&0\\0&\sd^{2}\eep V^{-1}\bb,
\ee
with an obvious  generalization to the case of several
\com s.
\EEN
\end{proposition}
\begin{proof} 1. This holds by  \eqr{ksys}.

2. This is an elementary computation.

3. One may check that the conditions of \cite{Van} hold, and thus the stability of the unique DFE is determined by the spectral radius of the \NGM
\[(\sd^{1} B_1 + \sd^{2} B_2) V^{-1},\] which coincides with the \PF\ eigenvalue.

 Note first  that when $B_i= \bb_i \va_i$, then  $B:=\sd^{1} B_1 + \sd^{2} B_2$ has rank $2$, and  so does the \NGM. Therefore, all its eigenvalues except at most two are $0$.

We show now that both the remaining eigenvalues satisfy also a $2 \times 2$ eigenvalue problem, extending the proof sketched in Remark 1.2.  We note first that  the corresponding eigenvectors to the right of the \NGM must be of the form $v=\kappa_1\bb_1+\kappa_2 \bb_2$. Plugging now this form yields the homogeneous $n \times 2$ system  for $\kappa_i, i=1, 2$:
\[\pp{(\sd^{1} \bb_1 \va_1 + \sd^{2} \bb_2 \va_2) V^{-1} -\mR I_n}\bep \bb_1,&\bb_2 \eep  \bep \kappa_1\\\kappa_2 \eep=0\]
 For non-zero solutions, the $n \times 2$ matrix multiplying $\bep \kappa_1\\\kappa_2 \eep$ must have rank one.

Putting further $a_{ij}= \va_i V^{-1} \bb_j$,  we may rewrite  the matrix
as  \[\bep \sd^{1} \bb_1 a_{11} + \sd^{2} \bb_2 a_{21},&\sd^{1} \bb_1 a_{12} + \sd^{2} \bb_2 a_{22} \eep -\mR \bep \bb_1,&\bb_2 \eep=\bep \bb_1,&\bb_2 \eep(M_D-\mR I_2).\]
Now if $\det(M_D-\mR I_2)=0,$ i.e. if \eqr{RMat} holds, then this matrix has rank one, and if $ \bb_1,\bb_2 $ are independent,  then this condition is also necessary. Finally, we note that when the matrix $M_D$ has two positive eigenvalues, then $\mR$ will be their maximum.\end{proof}

One example where this formula applies is the vector–host model in \cite[Sec. 4.5]{Van}.
\beR   The formula for the \brn, in  examples with two susceptible classes, involves sometimes square roots -- see for example \cite{heffernan2005perspectives,heesterbeek2007type}, \cite[(7,9)]{failure}, \cite[Sec. 6.1]{Van17}. 
 Other times it is $\max[\nR_1,\nR_2]$ where $\nR_1,\nR_2$ involve only one susceptible class, and are obtained by the rank one \brn\ formulas (this is sometimes referred to as competitive exclusion).
Our formula  \eqr{RMat} suggests that the difference between the two situations comes from the reduced \NGM\ (or the Diekmann kernel) being triangular or not.

\eeR

\section{Extending rank one ODE epidemic models via the generalized linear chain trick (GLCT) formalism, and the associated renewal models}\la{s:GLCT}
The rank one models discussed in this paper arise often via the so-called
generalized linear chain trick \cite{Hurtado19}. For example, the SIR-PH with one susceptible class arises by splitting the I individuals into several sub\com s, with in-between transitions governed by an $(\va,A)$
phase type distribution. This model may be viewed also as a SIR renewal model,  for which the time spent in the  infectious class is of type
$(\va,A)$. The  GLCT  formalism (GLCTF) is best explained via the example of the SEIR model.

\subsection{The SEIR-PH GLCT model} \label{sec:SEIR}
We recall here a   probabilistically interpretable SEIR-PH model, due to \cite{Hurtado19,Hurtado21}, and compute its Diekmann kernel and \bzn\ $\nR$.
Recall the classic SEIR model\begin{subequations}  \label{eq:SEIR}  \begin{align}
	\frac{dS}{dt} =& \; -\beta\,S\,I \label{eq:SEIRa} \\
	\frac{dE}{dt} =& \; \beta\,S\,I - r_E\,E \label{eq:SEIRb}\\
	\frac{dI}{dt} =& \; r_E\,E - r_I\,I \label{eq:SEIRc}\\
	\frac{dR}{dt} =& \; r_I\,I, \label{eq:SEIRd}
	\end{align}
\end{subequations}
where we have assumed  no demography, which renders  the computation of the semigroup simpler.
Assume now the latent period distribution is phase-type with parameters $\va_\mathbf{E}$ and $\mathbf{A_E}$, and the infectious period distribution is also phase-type, but with parameters $\va_\mathbf{I}$ and $\mathbf{A_I}$. Let $\vx=[E_1,\ldots]$ and $\vy=[I_1,\ldots]$ be the row vectors of the fraction of individuals in each of the exposed and infectious sub-states, respectively, and put $E=\sum E_j$ and
$I=\sum I_i$. An associated ODE model, which may be obtained either by introducing auxiliary unknowns $\vx,\vy$ into the SEIR model \eqr{eq:SEIR}, or directly by GLCTF, is: \begin{subequations} \label{eq:SEIRPT} \begin{align}
\frac{dS}{dt} =& \; -\beta\,S\,I  \\
\frac{d\vx}{dt} =& \; \beta\,S\,I\,\va_\mathbf{E} + \vx  {\mathbf{A_E}}\\
\frac{d\vy}{dt} =& \; \underbrace{\big(\vx \bff a_E \big)}_\text{destruction scalars} \cdotp \va_\mathbf{I} + \vy {\mathbf{A_I}}, \qu \bff a_E=-{\mathbf{A_E}\mathbf{1}} \\
\frac{dR}{dt} =& \; \overbrace{\vy \bff a_I}, \qu \bff a_I={-\mathbf{A_I}\mathbf{1}}.
\end{align} \end{subequations}
\beR The first fact to note in \eqr{eq:SEIRPT} is that the scalar transfer rates $r_E,r_I$ have been replaced by transfer matrices $\mathbf{A_E}, \mathbf{A_I}$, in their ``origine equations" \eqr{eq:SEIRb}, \eqr{eq:SEIRc}, \resp. The second  key fact to note that in the ``destination equations"
\eqr{eq:SEIRc}, \eqr{eq:SEIRd}, $r_E,r_I$ have been replaced by the product of ``destruction scalars" (ending with $\bff a_E, \bff a_I$) by  ``rebirth vectors" ($\va_\mathbf{I}$ for the first, and $1$ for the second). This is the matrix expression of the usual balance between out-flows and in-flows, which takes into account the different dimensionality of the origines and destinations. \eeR

This is an  $(A,B)$ Arino-Brauer epidemic model with one susceptible class and  parameters
 \begin{align*}
 -V=\mathbf{A}= \bep
 {\mathbf{A_E}} & \bff a_E \va_I \\
 0 & {\mathbf{A_I}}
\eep, \va =\bep \va_E&0\eep,
  \bb=\bep 0 \\ \beta \bff 1\eep.
\end{align*}
\beR Researchers familiar with the interpretations of $\bff a_E$ and $ \va_I$ as death rates in E states and birth rates in I states will note the intuition behind the formulas, as well as the fact that by grouping together E, I, in an I class, the SEIR-PH may be viewed also as a SIR-PH. Finally, note that the generalization to heterogeneous infectivity rates $\bb =\bep \beta_ 1\\\beta_ 2\\\vdots\eep$ is immediate.
\eeR
By the \cite{Arino} formula, the  \brn\ is
\be{RSE}
\mR= \bep \va_E&0\eep \bep
 -{\mathbf{A_E}}^{-1} & {\mathbf{A_E}}^{-1} \bff a_E \va_I {\mathbf{A_I}}^{-1} \\
 0 & -{\mathbf{A_I}}^{-1}
\eep \bep 0 \\ \beta \bff 1\eep=-\beta \va_I {\mathbf{A_I}}^{-1} \bff 1=\beta \E[\tau_E +\tau_I],\ee
where $\tau_E, \tau_I$ denote the distributions of the total times spent in exposed and infectious classes, \resp, and $\E$ denotes mathematical expectation.  The semigroup is quasi-explicit, given by
\be{seHu} e^{t \mathbf{A}}=\bep e^{t\mathbf{A_E}}&ILT \pp{(s Id-\mathbf{A_E})^{-1}\bff a_E \va_I (s Id-\mathbf{A_I})^{-1}} \\
\mathbf{0}&e^{t\mathbf{A_I}},\eep
\ee
where $ILT$ denotes inverse Laplace transform.
\begin{example}
For example, suppose
\[\mathbf{A_E}=\bep -\eta_E &\eta_E\\0&-r_E\eep, \bff a_E= \bep 0\\r_E\eep, \mathbf{A_I}=\bep -\eta_I &\eta_I\\0&-r_I\eep, \bff a_I= \bep 0\\r_I\eep, \va_I=(1,0).\]

Then, $\mathbf{A}=\bep -\eta_E &\eta_E&0&0\\0&-r_E&r_E&0\\0&0&-\eta_I &\eta_I\\0&0&0&-r_I\eep$, reflecting the Erlangization of the $E$ and $I$ classes.

The semigroup $e^{t A}$ is explicit, for example the LT transform in its NE corner is
\[\fr 1 {\eta _i+s}\frac{r_e}{r_e+s }
\left(
\begin{array}{cc}
 \frac{\eta _e }{ \eta _e+s}
 & \frac{\eta _e }{ \eta _e+s} \frac{ \eta _i}{  r_i+s } \\
  1& \frac{ \eta _i}{ r_i+s } \\
\end{array}
\right).\]
\end{example}

\subsection{A seven compartments epidemic model of Covid-19 inspired by  \cite{torres20,domo22}}\la{s:cov}

The following seven compartments epidemic model of Covid-19 includes besides the omnipresent $\s,i$, also additional   super-spreaders,  hospitalized, recovery and  fatality classes denoted by $p,h,r,f$, respectively. It is given by:

\be{covid} \bc
\s'(t)=   -\beta_ i \s(t) i(t) -\beta_ h \s(t) \h(t)- \beta_ p \s(t) \p(t)\\
\bep \e'(t)&i'(t)&\p'(t)&\h'(t)\eep = \bep \e(t)&i(t)&\p(t)&\h(t)\eep
\Biggl[\s(t) \left(
\begin{array}{cccc}
 0 & 0 & 0 & 0 \\
 \beta _i & 0 & 0 & 0 \\
 \beta _p & 0 & 0 & 0 \\
 \beta _h & 0 & 0 & 0 \\
\end{array}
\right)+\\
  \left(
\begin{array}{cccc}
 -\gamma_e-e_r & e_i & e_p & 0 \\
 0 & -i_h-i_r-\delta _i & 0 & i_h \\
 0 & 0 & -p_h-i_r-\delta _p & p_h \\
 0 & 0 & 0 & -\g_h-\delta _h \\
\end{array}
\right)\Biggr]
\\
\bep r'(t)& \f'(t)\eep= \bep \e(t)&i(t)&\p(t)&\h(t)\eep \left(
\begin{array}{cccc}
 e_r & i_r & p_r & \g_h \\
 0 & \delta _i & \delta _p & \delta _h \\
\end{array}
\right)^t \ec.
\ee

\beR We have removed from the original model a class called asymptomatic, because that class was not allowed to produce new infections \cite[Fig 1]{torres20}, as it should have under the usual definition of asymptomatics -- see for example \cite{Ansumali},  and also did not allow recovery, which was probably a typo.
\eeR

This is an  \AB\ with  parameters
 \begin{align*}
 V=  \left(
\begin{array}{cccc}
 \gamma_e & -e_i & -e_p & 0 \\
 0 & i_h+i_r+\delta _i & 0 & -i_h \\
 0 & 0 & \p_h+p_r+\delta _p & -\p_h \\
 0 & 0 & 0 & \g_h+\delta _h \\
\end{array}
\right), \va =\bep 1&0&0&0\eep,
  \bb=\bep 0 \\\beta_ i\\\beta_ p\\ \beta_h\eep.
\end{align*}

The Arino \& al. formula yields
$
\nR=e_p \frac{\beta _h p_h+\beta _p \left(\delta _h+\g_h\right)}{\gamma _e  \left(p_h+\delta _p+p_r\right) \left(\delta _h+\g_h\right)}+ e_i \frac{\beta _h i_h+\beta _i \left(\delta _h+
\g_h\right)}{\gamma _e \left(i_h+\delta _i+i_r\right) \left(\delta _h+\g_h\right)}.$

\figu{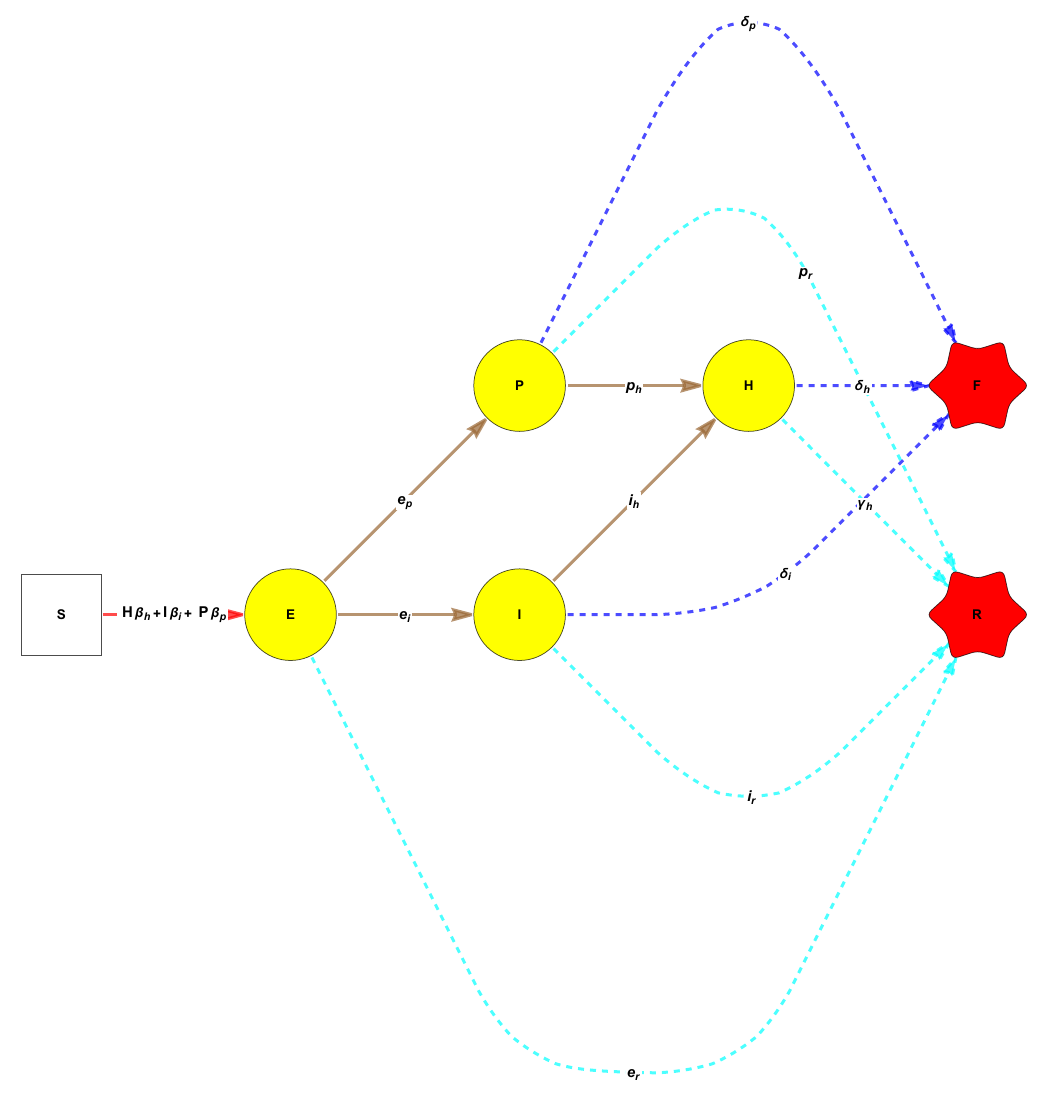}{ Chart flow of the Covid-19 model \eqr{covid}. }{0.9}

\beR \cite{domo22} argue that the transitions from $i$ and $\p$ to $\h$ should be modeled via DD equations, obtained by replacing $i,\p$ by convolutions with exponential densities, and obtain associated ODE's with two extra unknowns.  We note here that their system loses the conservation
of mass of the \cite{torres20} system, but that this  may be maintained by applying GLCT, as follows.
\BEN \im
Start by introducing new auxiliary variables \[\T i(t)=\int_0^t  \eta_i e^{- \eta_i(t-\tau)}\,i(\tau) d\tau, \; \T \p(t)=\int_0^t  \eta_p e^{- \eta_p(t-\tau)}\,\p(\tau) d\tau, \]
(note that these converge to $e(t),p(t)$ when $\eta_p, \eta_i\to\I$).
\im Replace $i,\p$ in the equation for $\h$ by $\T i,\T \p$.
\im Add the differential equations
$\frac{d\T i}{dt} =\; \eta_i\,i - \eta_i \,\T i \;\;
\frac{d\T \p}{dt} =\eta_p\,\p - \eta_p \,\T \p  $
to the ODE system.

\EEN
Finally, we arrive, for the 6 extended disease equations, to:
\begin{align*}
\bc
\e'(t)=\beta_i \s(t) i(t)+\beta_p \s(t) \p(t) +\beta_h \s(t) \h(t)-\g_e\e(t),\\
i'(t)= e_i \e(t)-(i_h+i_r+\delta_i)i(t),\\
\T i'(t)=\eta_i i(t) -\eta_i \T i(t),\\
\p'(t)=e_p \e(t)-(p_h+p_r+\delta_p) \p(t),\\
\T \p'(t)=\eta_p \p(t)-\eta_p \T \p(t),\\
\h'(t)=i_h \T i(t)+p_h \T \p(t)-(\g_h+\delta_h)\h(t)
\ec
\end{align*}

This is an \AB\ with disease variables

$\bep \e(t)&i(t)&\T i(t)&\p(t)&\T \p(t)&\h(t)\eep$ and  parameters $\va =\bep 1&0&0&0&0&0\eep,$
 \begin{align*}
 V=  \left(
\begin{array}{cccccc}
 \gamma_e & -e_i & 0& -e_p & 0 & 0 \\
 0 & i_h+i_r+\delta _i & -\eta_i &0&0&0 \\
  0 & 0 & \eta_i&0&0 & -i_h \\
 0 & 0 & 0&p_h+p_r+\delta _p & -\eta_p &0 \\
 0 & 0 & 0&0 & \eta_p &-p_h \\
 0 & 0 & 0 &0&0& \g_h+\delta _h \\
\end{array}
\right),
  \bb=\bep 0 \\ \beta_ i\\ 0\\\beta_ p\\ 0\\ \beta_h\eep.
\end{align*}

We must assume that $i_h <\eta_i, p_h <\eta_p$, so that $A$ is a sub-generator. It follows than by the formula of Arino $\&$ al. that: \[
\nR=e_p \frac{\beta _h p_h+\beta _p \left(\delta _h+\g_h\right)}{\gamma _e \left(\delta _h+\g_h\right) \left(\delta _p+p_h+p_r\right)}+ e_i \frac{\beta _h i_h+\beta _i \left(\delta _h+\g_h\right)}{\gamma _e \left(\delta _h+\g_h\right) \left(\delta _i+ i_h+i_r\right)},\]
which is precisely the same $\nR$ as before the extension.
\eeR
\beR There does not seem to be  much practical estimation work
of phase-type distributions fitting real epidemic data -- see though \cite{Hurtado19,kim2019global}. \eeR
In the following subsection we present one example for which we plan to undertake  such work in the future.

\section{Conclusions and further work} \la{s:conc}
Solving the exercise of \cite{Breda} revealed that the  \AB\ with $B$ of rank one have the remarkable property of having a natural associated ``age of infection kernel" which implies a very simple formula for
$\nR$.\fn[4]{Thus, these deterministic models have also  one foot in the stochastic  world, which reveals itself when all the infectious equations are grouped into one equation.} This continues to be true for models with several susceptible classes, and is
 of considerable interest for epidemic models structured by the age of the individuals.

{We  prove now that \eqr{kerMat} has a stationary point with $\T i\neq (0,0)$ iff $\mR$ defined in \eqr{RMat} is bigger than $1$,  that  this stationary point is locally stable}.

This proof reveals also that $\mR$ may also be obtained as the spectral radius of the integral of the matrix kernel
$K(\tau)=Diag(\sd)  a(\tau), \sd=(\sd^{1},\sd^{2})$.  We will call the matrix $K(\tau)$ above Diekmann  matrix kernel, in reference
to \cite[(5.9)]{Diek18}, where such matrices seem to have appeared for the first time. Note that in our situation $K(\tau)$ is explicit.

 Another  question worth further research is whether an age of infection kernel may  be associated to
  \AB\ with one susceptible  class, but with a matrix $B$ of rank bigger than $1$.

  The results of this paper suggest an interesting alternative to the classical statistical  approaches to \ME, which usually start by postulating an epidemiologic model, and then estimate its parameters.
  The alternative consists  in accepting that the  model is not fully known. and estimate instead  an age of infection kernel. Subsequently,
  a good matrix exponential approximation of the kernel will translate directly into an epidemiologic model, to be confronted with the currently accepted ones.

{\bf Aknowledgements}. We thank Tyler Cassidy, Odo Diekmann, and James Watmough for useful remarks.

\bibliographystyle{amsalpha}
\bibliography{Pare40}


\end{document}